\documentclass[letterpaper, 10 pt, conference]{ieeeconf}  

\pdfoutput=1 
\IEEEoverridecommandlockouts                              

\usepackage{cite}
\usepackage{graphicx}

\def\BibTeX{{\rm B\kern-.05em{\sc i\kern-.025em b}\kern-.08em
    T\kern-.1667em\lower.7ex\hbox{E}\kern-.125emX}}
    
\usepackage{url}
\usepackage[cmex10]{amsmath}
\usepackage{amsfonts}
\usepackage{amssymb}
\usepackage{mathrsfs}
\usepackage{dsfont}
\usepackage{braket}
\usepackage{subfig}
\usepackage{tikz}
\usetikzlibrary{matrix}
\usepackage{blkarray}
\usepackage{multirow}
\usepackage{footnote}
\usepackage{algorithmicx}
\usepackage[ruled]{algorithm}
\usepackage{algpseudocode}
\usepackage{fancyhdr}  				
\usepackage{graphicx,float,wrapfig}				
\makesavenoteenv{tabular}
\usepackage[font=small,belowskip=-9pt,aboveskip=4pt,labelsep=period,justification=justified]{caption}
\makeatletter

\renewcommand*\env@matrix[1][*\c@MaxMatrixCols c]{%
  \hskip -\arraycolsep
  \let\@ifnextchar\new@ifnextchar
  \array{#1}}
\makeatother
\usepackage{flushend}
\usepackage{color}
\usepackage{accents}

\newcommand{\sw}{\text{sw}}
\newcommand{\dc}{\text{dc}}

\graphicspath{{./images/}}

\newcommand{\ty}{\infty}
\newcommand{\mc}[1]{\mathcal{#1}}

\newcommand{\mb}[1]{\mathbb{#1}}
\newcommand{\Ra}{\;\;\Rightarrow\;\;}

\newcommand{\ra}{\rightarrow}

\newcommand{\txt}[1]{\text{#1}}

\newcommand{\mat}[1]{\begin{matrix}#1\end{matrix}} 
\newcommand{\pa}[1]{\left(#1\right)} 
\newcommand{\br}[1]{\left[#1\right]} 
\newcommand{\pmt}[1]{\pa{\mat{#1}}} 

\newcommand{\q}{\quad}
\newcommand{\s}{&}

\newcommand{\mbf}[1]{\mathbf{#1}}
\newcommand{\teq}{\triangleq}

\newtheorem{proposition}{Proposition}

\title{ \LARGE \bf Nonlinear Model Predictive Control of Permanent Magnet Synchronous Generators in DC Microgrids}

\author{Luis Herrera, Chad Miller, and Bang-Hung Tsao
	\thanks{Luis Herrera is with the Department of Electrical Engineering, University at Buffalo, Buffalo, NY. Chad Miller is with the Air Force Research Laboratory, WPAFB, OH. Bang-Hung Tsao is with the University of Dayton Research Institute, Dayton, OH. Corresponding author: L. Herrera {\tt\small lcherrer@buffalo.edu}}%
}

\begin{document}



\maketitle
\thispagestyle{empty}
\pagestyle{empty}

\begin{abstract}
A new strategy is proposed to control interior permanent magnet generators in  dc microgrids interfaced through an active rectifier. The controller design is based on the decomposition of the system dynamics into slow and fast modes using singular perturbation theory. An inner current controller is developed based on output regulation techniques and an outer voltage controller is proposed using  Nonlinear Model Predictive Control (NMPC). The NMPC regulates the dc bus voltage and  minimizes the ac side losses. Simulation results are then presented based on realistic conditions for aircraft power systems.


\end{abstract}


\section{Introduction}

Electric machines play a fundamental role in the development of dc microgrids, with applications in the transportation industry. In electric vehicles, Permanent Magnet Synchronous Machines (PMSM) are a popular choice for the primary motor/generator \cite{Chau}. In the More Electric Aircraft (MEA), these machines can be used for generation  and motoring applications (e.g. actuators, propellers, etc.) \cite{Giangrande, Gao1}. 

Control techniques for PMSMs used in motor drives generally ensure their optimal operation (in terms of efficiency) using techniques such as Maximum Torque per Amp (MTPA) and Maximum Torque per Volt (MTPV) \cite{nam2018ac}. These optimal conditions are relatively straightforward to implement in Surface Mounted PMSM (SPMSM), since the torque production only involves the permanent magnet and the q-axis current. Therefore, for a SPMSM in motoring mode, the q-axis current is  used primarily to track a certain speed or torque reference during normal operation. In generator mode, this same current can be used to regulate the dc bus voltage \cite{Fan2018}.

However, the controllers  of Interior PMSMs (IPMSM) based motors/generators do not always operate optimally. The main reason is that IPMSM machines can produce torque through both its permanent magnets and through the reluctance torque, due to the saliency of the rotor. Since the latter utilizes both d and q axis currents, when the same strategy as SPMSM is used for IPMSM, the reluctance torque is not optimally used and the generator/motor is operated at a lower power factor (increasing ac side losses). For example, in \cite{Gao1, Fan2018, Dehghani, Miao, Gao2016,  Bozhko2017, Clements2009,tripathi2016optimum}, the q axis current is used to control the dc bus voltage (generator) and speed/torque (motoring), irrespective of the type of machine used (SPMSM or IPMSM).

Most controllers for dc/ac and dc/dc converters employ a two loop strategy: inner current control and outer voltage or speed/torque control \cite{wang2014modeling, HerreraNCS, vasquez2012modeling}, and their stability analysis is typically presented using linearization  techniques such as root locus \cite{pogaku2007modeling}. This particular control structure owes its development to the nature of the physical system, composed of both fast and slow states. However, the (nonlinear) stability analysis and controller design for these types of controllers exploiting these fast/slow time constants has not been conducted. Nevertheless, singular perturbation techniques have been employed for power electronics and motor drives \cite{kimball2008singular, umbria2014three,kokotovic1999singular}. However, this type of control design does not generally follow an inner/outer loops and instead uses a composite control, i.e. a summation of two terms: the slow and fast components.

In this paper, we analyze the dynamics of PMSMs based generators for dc power systems using singular perturbation techniques and develop a  controller which maintains the  existing inner/outer loop control structure typically used in power electronics. In Section II, the overall generator dynamics with an active rectifier is presented along with an overview of the control procedure. In Section III, the inner current controller is developed using output regulation theory to track the desired reference. In Section IV, the outer controller is proposed using Nonlinear Model Predictive Control (NMPC) to achieve both voltage regulation and optimal operation of the machine. Simulation results are presented in Section V based on a BMW i3 IPMSMs in rectification mode (generator). Lastly, conclusion and future work are discussed.

The following notation is used throughout this paper.    For a general matrix $M\in\mb{R}^{n\times m}$, its $(i,\;j)$ element is denoted as $M_{(i,j)}$.  The set of complex numbers with negative real part is denoted as $C^-$.

\begin{figure}[!b]
	\begin{center}
		\includegraphics[width=0.4\textwidth]{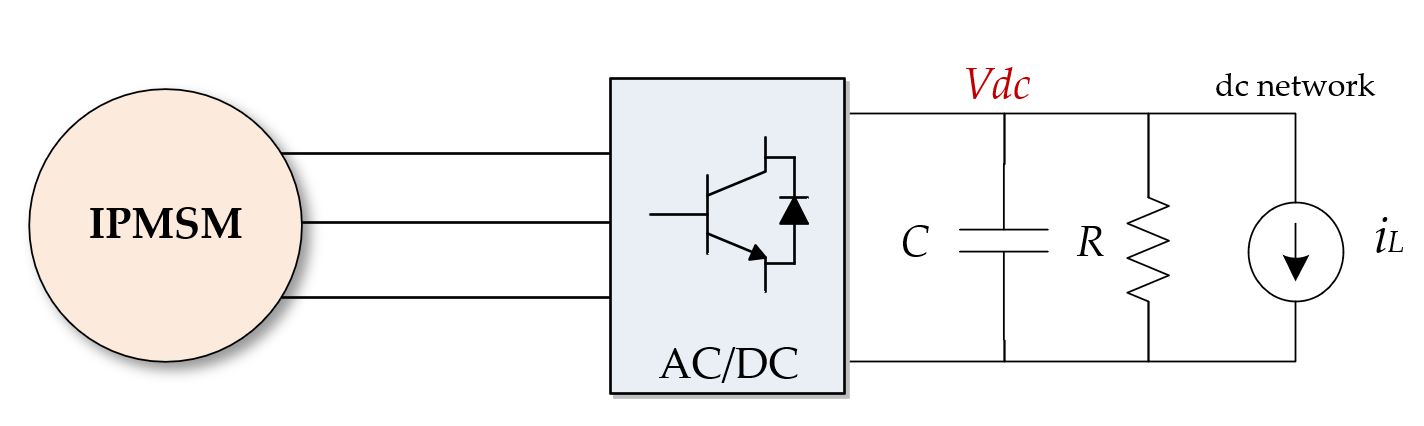}\vspace{0pt}
		\caption{PMSM based generator for dc microgrids.} 
		\label{fig:genair} 
	\end{center}
\end{figure}

\newcommand{\meas}{0.14}

\renewcommand{\arraystretch}{1.3}

\section{PMSM Based Generator}
An overview of a PMSM generator (PMSG) for dc microgrids is shown in Fig. \ref{fig:genair}. The overall dynamics are composed of the ac side (PMSM) and dc side (capacitor) system. These systems can effectively be decomposed into fast and slow modes.

\subsection{System Dynamics} 
The  dynamics of a PMSG with an active rectifier in rotor reference frame are given as follows \cite{nam2018ac, Gao2016}:
\renewcommand{\arraystretch}{1.75}
\begin{align}\label{eq:dcside}
&\txt{(dc)}\left\{\begin{array}{ll}
\dot{v}_{dc} \s= -\frac{1}{RC}v_{dc} + \frac{3}{2C}\frac{1}{v_{dc}}\pa{v_di_d+v_qi_q} -\frac{1}{C}i_L
\end{array}\right. \\\label{eq:acside}
&\txt{(ac)}\left\{\begin{array}{ll}
\dot{i}_d &= \frac{-R_s}{L_d}i_d+\omega_r \frac{L_q}{L_d}i_q + \frac{1}{L_d}v_d \\
\dot{i}_q &= \frac{-R_s}{L_q}i_q -\omega_r\frac{L_d}{L_q}i_d -\frac{\omega_r}{L_q}\lambda_m + \frac{1}{L_q}v_q
\end{array}\right.
\end{align}
where $i_{d}$ and $i_{q}$ are the d and q axis current respectively, $L_d$ and $L_q$ are the inductances in the respective axis, $R_s$ is the stator resistance, $\omega_r$ is the rotor electrical frequency, $\lambda_m$ is the permanent magnet flux linkage, $C$ is the dc side capacitance, $R$ is the parallel dc side resistance, and $i_L$ is the dc side load current. The overall system was derived using the standard dq-transformation shown in the appendix.

The inputs to \eqref{eq:dcside}-\eqref{eq:acside}, $v_d$ and $v_q$, can be written in terms of the modulation indices $d_d,\;d_q\in\br{-1,\;1}$:
\begin{align}\label{eq:modinx}
\begin{array}{ll}
v_d &= d_d\frac{v_{dc}}{2}\\
v_q &= d_q\frac{v_{dc}}{2}
\end{array}
\end{align} 
based on sine PWM, coupling the dc voltage to the ac currents.

The system model, \eqref{eq:dcside}-\eqref{eq:acside}, can then be decomposed into fast and slow modes and written using singular perturbation theory as follows \cite{kokotovic1999singular}:
\renewcommand{\arraystretch}{1.45}
\begin{align}
\label{eq:SPslow}
\dot{x} &= f(x,\;z,\;u,\;d)\\ \label{eq:SPfast}
\mu\dot{z} &= g(x,\;z,\;u, \; d)
\end{align}
where $0<\mu\ll1$, $x\teq v_C$  and $z\teq\pa{i_d,\;i_q}^T$ are the slow and fast modes respectively, $u = \pa{d_d,\;d_q}^T$ are the inputs, and  $d\teq i_L$ is the disturbance.

\subsection{Overall Controller Design}
The goal of a PMSG controller is  to regulate the dc bus voltage. Typical controller design using singular perturbation theory decomposes the inputs into slow and fast components as $u=u_s+u_f$, generally known as composite control \cite{kokotovic1999singular}. However, following existing approaches for control of electric machines \cite{Gao2016} and power electronics \cite{vasquez2012modeling}, the controller will be developed as follows:
\begin{itemize}
	\item The fast modes are regulated through $u$ to follow a desired reference, i.e. $z\ra z^*$
	\item The slow modes are controlled through  $z^*$ to follow a certain reference, i.e. $x\ra x^*$
\end{itemize}
\subsubsection{Inner Loop}
For the fast mode controller design, the slow modes, $x$, are assumed to be constant, i.e. $x= \bar{x}$, and thus \eqref{eq:SPfast} can be written as:
\begin{align}
\mu\dot{z} = g(\bar{x},\;z,\;u,\;d) = \tilde{g}(z,\; u,\;d)
\end{align}
where $u$ is designed by a static or dynamic controller to ensure fast regulation: $z\ra z^*$. 

\subsubsection{Outer Loop}
The outer/slow controller assumes  that the dynamics of the closed loop fast subsystem are instantaneous:
\begin{align}\label{eq:fastsp}
0 = g(x,\;z,\;u,\;d)
\end{align}
and  $u$ can be obtained from \eqref{eq:fastsp} and  written as a function of the slow and fast modes, i.e.  $u=p(x,\;z,\;d)$. The slow subsystem then becomes:
\begin{align}
\dot{x} = f(x,\;z,\;p(x,\;z,\;d),\;d) = \tilde{f}(x, \;z^*, \;d)
\end{align}
with the new input $z^*$.

\begin{figure}[!t]
	\begin{center}
		\includegraphics[width=0.5\textwidth]{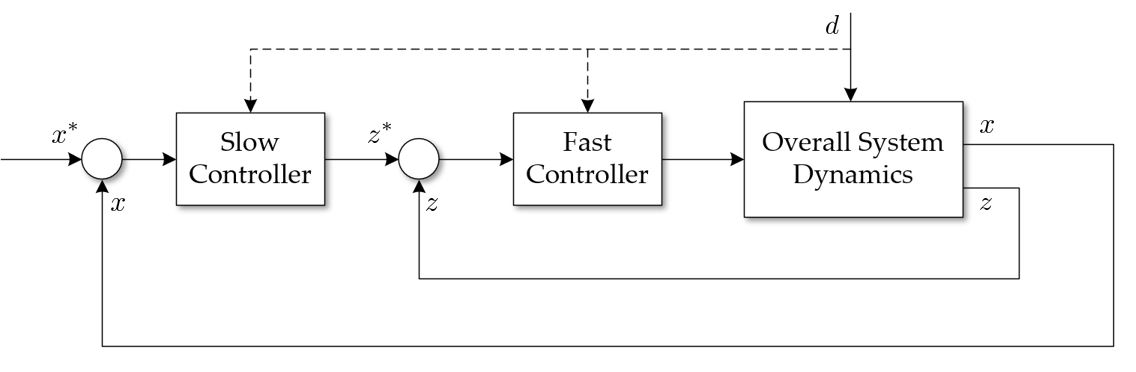}\vspace{0pt}
		\caption{Overview of the proposed controller design for singular perturbation systems.  The dashed lines depends on the type of controller used and sensor availability.} 
		\label{fig:overallctrl} 
	\end{center}
\end{figure}
Fig. \ref{fig:overallctrl} presents an overview of the proposed controller design. As can be inferred from this figure,  the closed loop dynamics for the fast subsystem $z$ need to be much faster than $x$.   Therefore, fast regulation of $z\ra z^*$ is a crucial requirement. 

\section{Fast Inner  Current Regulator}
In this section, the controller design for the fast subsystem defined by \eqref{eq:acside} is presented. Following the  procedure outlined in Section 1B, the slow mode (dc bus voltage)  is assumed to be constant, i.e. $v_{dc}=\bar{v}_{dc}$. Therefore, the inputs/modulation indices, $d_d$ and $d_q$, can  be re-written in terms of the voltages $v_d\;\txt{and}\;v_q$ respectively based on \eqref{eq:modinx}.  In this case, \eqref{eq:acside} becomes a linear state space system. 

We consider standard decoupling techniques for inverters \cite{Gao2016} by defining new inputs, $\tilde{v}_d,\;\tilde{v}_q$, as follows:
\begin{align}
\begin{array}{ll}\label{eq:ctrlmods}
v_d &= \tilde{v}_d - \omega_r L_q i_q \\
v_q &= \tilde{v}_q + \omega_r L_d i_d + \omega_r \lambda_m
\end{array}
\end{align}
Plugging \eqref{eq:ctrlmods} into \eqref{eq:acside} we can obtain the following:
\begin{align}
\dot{i}_d &= \tfrac{-R_s}{L_d}i_d +\tfrac{1}{L_d}\tilde{v}_d \label{eq:ddecoup}\\
\dot{i}_q &= \tfrac{-R_s}{L_q}i_q + \tfrac{1}{L_q}\tilde{v}_q \label{eq:qdecoup}
\end{align}
Therefore, each current controller can be designed independently. The controllers' goals are to regulate $z\ra z^*$ or $i_d\ra i_d^*$ and $i_q\ra i_q^*$.
\renewcommand{\arraystretch}{1.25}
\subsection{$D$ Axis Current Control}
For the $d$ axis current controller, we assume that the overall system, including the reference $i_d^*$ (constant), is defined by the following linear dynamics: 
\begin{align}
\begin{array}{ll}
\dot{i}_d &= \tfrac{-R_s}{L_d}i_d +\tfrac{1}{L_d}\tilde{v}_d \\
\dot{i}_d^* &= 0 \\
e_d &= -i_d + i_d^*
\end{array}\Ra
\begin{array}{ll}
\dot{x}_d &= A_d x_d + B_d u_d \\
\dot{x}_d^* &= S_d x_d^* \\
e_d &= C_dx_d + Q_dx_d^*
\end{array}
\end{align}
The controller, $u_d$, is designed using output regulation techniques \cite{Francis}:
\begin{align}\label{eq:ctrld}
 u_d = K_d\xi_d + T_dx_d^*
\end{align}
where $\xi_d$ is an estimate of $x_d$ (e.g. using a Luenberguer or Kalman filter), $K_d$ is designed such that $\sigma(A_d+B_dK_d)\subset\mb{C}^{-}$ and $T_d$ is a feed forward gain satisfying:
\begin{align}\label{eq:outregd}
\begin{array}{rl}
A_d\Pi_d+B_d(K_d\Pi_d +T_d) &= \Pi_d S_d\\
C_d\Pi_d+Q_d &= 0.
\end{array}
\end{align}

\begin{proposition}
	The control law,  $u_d = K_d\xi_d + T_dx_d^*$, satisfying \eqref{eq:outregd} and $\sigma(A_d+B_dK_d)\subset\mathbb{C}^-$ ensures $i_d\ra i_d^*.$
\end{proposition}
\noindent The proof follows standard arguments of output regulation theory \cite{Francis}.
\subsection{$Q$ Axis Current Control}
A similar procedure is followed for the $q$-axis current regulator. The overall dynamics are as follows:
\begin{align}
\begin{array}{ll}
\dot{i}_q &= \tfrac{-R_s}{L_q}i_q + \tfrac{1}{L_q}\tilde{v}_q  \\
\dot{i}_q^* &= 0 \\
e_q &= -i_q + i_q^*
\end{array}\Ra
\begin{array}{ll}
\dot{x}_q &= A_q x_q + B_q u_q \\
\dot{x}_q^* &= S_q x_q^* \\
e_q &= C_qx_d + Q_qx_q^*
\end{array}
\end{align}
and the goal is to ensure that $i_q\ra i_q^*$ or $e_q\ra 0$ as $t\ra \ty$. The controller law is $u_q$ is given in a similar form:
\begin{align}
u_q = K_q\xi_q +T_q x_q^*
\end{align}
where $\xi_q$ is an of the current $i_q$ through a linear observer. The feedback matrix $K_q$ is designed to ensure that $\sigma(A_q+B_qK_q)\subset\mb{C}^{-}$ and $T_q$ satisfies the same full information regulator equations:
\begin{align}\label{eq:outregq}
\begin{array}{rl}
A_q\Pi_q+B_q(K_q\Pi_q +T_q) &= \Pi_q S_q\\
C_q\Pi_q+Q_q &= 0.
\end{array}
\end{align}
\begin{proposition}
	The control law $u_q = K_q\xi_q+T_qx_q^*$ satisfying $\sigma(A_q+B_qK_q)\subset\mb{C}^-$ and \eqref{eq:outregq} ensures $i_q\ra i_q^*$ as $t\ra \ty$.
\end{proposition}

Lastly, the feedback matrices $K_d$ and $K_q$ should be optimized carefully in order to guarantee that fast regulation of the currents $i_d$ and $i_q$.  Semi-definite Programming (SDP) techniques are used in the case study section for the tuning of the controller gains \cite{HerreraTSG}.

\section{ NMPC Based DC Voltage Control - Slow Subsystem }
NMPC is used for the controller design of the slow subsystem, $v_\dc$, in order to provide the references $i_d^*$, $i_q^*$ to the inner current control. In this case, it is assumed that the fast dynamics are instantaneous, i.e. $\mu\ra0$ in \eqref{eq:SPfast}. Based on \eqref{eq:acside}, the left hand side is simplified as:
\begin{align}
\begin{array}{ll}
0 &= \frac{-R_s}{L_d}i_d+\omega_r \frac{L_q}{L_d}i_q + \frac{1}{L_d}v_d \\
0 &= \frac{-R_s}{L_q}i_q -\omega_r\frac{L_d}{L_q}i_d -\frac{\omega_r}{L_q}\lambda_m + \frac{1}{L_q}v_q
\end{array}
\end{align} 
and $v_d$ and $v_q$ can be obtained as:
\begin{align}\label{eq:inputsimp}
\begin{array}{ll}
v_d \s = -r_s i_d+\omega_r L_q i_q\\
v_q \s = -r_s i_q-\omega_r L_di_d-\omega_r \lambda_m
\end{array}
\end{align}
Plugging \eqref{eq:inputsimp} into the dc voltage dynamics \eqref{eq:dcside}, the following nonlinear system is derived:
\begin{align}\label{eq:vdcsimp}
\begin{array}{ll}
\s\dot{v}_{dc} = -\frac{1}{RC}v_{dc}  -\frac{1}{C}i_L \\
\s+ \frac{3}{2C}\frac{1}{v_{dc}}\pa{-r_s(i_d^2+i_q^2)+\omega_r(L_q-L_d)i_qi_d-\omega_r\lambda_mi_q}
\end{array}
\end{align}
Notice that the inputs in this case are now $i_d=i_d^*$ and $i_q=i_q^*$, i.e. the references for the inner loop controller.  In addition, the new model \eqref{eq:vdcsimp} is nonlinear due to the second degree terms in the inputs and the reciprocal of the state term ($1/v_{dc}$).  

The goal of the slow subsystem controller is to regulate the dc bus voltage to a certain reference, $v_{dc}^*$, while at the same time reducing losses and satisfying constraints associated with the voltage boundaries (e.g. see MIL-STD-704F \cite{704fstd}) and the physical limits of  the PMSM (current and voltage). 

\subsection{Optimal Operation and Constraints}
Since only active power is consumed by the dc side of a PMSG, the ac side currents should be controlled as to provide only active power whenever possible (i.e. unity power factor). The torque produced by the PMSG is defined as follows:
\begin{align}
	T_e &= \frac{3}{2}\frac{P}{2}\pa{\lambda_m i_q + (L_d-L_q)i_qi_d}\q\txt{(Nm)}
\end{align}
where $P$ is the number of poles. Therefore, the electrical power, at the ac/mechanical side, can be obtained from the previous equation using the torque/power relation:
\begin{align}
P_e = T_e\omega_m &= \tfrac{3}{2}\tfrac{P}{2}\omega_m\pa{\lambda_m i_q + (L_d-L_q)i_qi_d} \\
\Ra\;P_e &= \tfrac{3}{2}\omega_r\pa{\lambda_m i_q + (L_d-L_q)i_qi_d} \label{eq:acpower}
\end{align}
where $\omega_m$ is the rotor mechanical speed (rad/sec) and the last equation is obtained from $\omega_r = \frac{P}{2}\omega_m$.

For generation mode, the electrical power is decided only by the dc load. However, since $P_e$ is a function of both $i_d$ and $i_q$, there are multiple solutions to \eqref{eq:acpower}.  The optimal solution minimizes the rms (or peak) of the ac side currents, i.e. providing only active power whenever possible. 

The constraints for a PMSM typically involve current and voltage limits. These can be written as follows:
\begin{align}
i_d^2+i_q^2 \s\leq I_\txt{peak}^2\\
v_d^2+v_q^2 \s\leq V_\txt{peak}^2 \label{eq:voltlimits}
\end{align}
Plugging \eqref{eq:inputsimp} into \eqref{eq:voltlimits} and assuming $r_s\approx0$, we can rewrite the voltage constraints in terms of $dq$ currents:
\begin{align}\label{eq:voltineq}
\pa{\omega_rL_qi_q}^2+\pa{\omega_rL_di_d+\omega_r\lambda_m}^2\leq V_\txt{peak}^2 = \pa{\frac{v_\dc}{2}}^2
\end{align}
The equality $V_\txt{peak}^2 = \pa{\frac{v_\txt{dc}}{2}}^2$ is based on sine PWM as shown in \eqref{eq:modinx} for $d_q=d_d=1$.

\renewcommand{\arraystretch}{1.35}
Finally, the optimal operation of the PMSM based generator for a fixed dc load power, $P_e$, is a solution of the following optimization problem:
\begin{align}\renewcommand{\arraystretch}{1.5}\label{eq:optstatic}
\begin{array}{l}
	\displaystyle \min_{i_d,\;i_q} \q i_d^2+i_q^2 \\
	\txt{s.t.} \\
	\q \tfrac{3}{2}\omega_r\pa{\lambda_m i_q + (L_d-L_q)i_qi_d} = P_e \\
	\q i_d^2+i_q^2\leq I_\txt{peak}^2\\
	\q \pa{\omega_rL_qi_q}^2+\pa{\omega_rL_di_d+\omega_r\lambda_m}^2\leq  \pa{\frac{v_\dc}{2}}^2
\end{array}
\end{align}
During steady state operation, the slow side controller should satisfy \eqref{eq:optstatic}. Of particular importance are the non-trivial solutions for \eqref{eq:optstatic}, contained in the interior of the following set:
\begin{align}
	\begin{array}{ll}
		\mc{E} = \s\bigg\{(i_d,\;i_q)^T\in\mb{R}^2\;\big|  i_d^2+i_q^2\leq I_\txt{peak}^2, \\
		\s \pa{\omega_rL_qi_q}^2+\pa{\omega_rL_di_d+\omega_r\lambda_m}^2\leq  \pa{\frac{v_\dc}{2}}^2\bigg\}
	\end{array}
\end{align}
i.e. when the inequalities in \eqref{eq:optstatic} are non-binding.  For this case, it is possible to supply only active power from the generator, hence minimizing the ac currents.

\subsection{NMPC Formulation}
We consider a NMPC controller for dc bus voltage regulation and optimal operation of the PMSG. To ensure convergence to the desired reference voltage, we expand \eqref{eq:vdcsimp} by an integral term as follows:
\begin{align}\label{eq:overalldc}
\begin{array}{l}
\begin{array}{ll}
\s\dot{v}_{dc} = -\frac{1}{RC}v_{dc}  -\frac{1}{C}i_L \\
\s+ \frac{3}{2C}\frac{1}{v_{dc}}\pa{-r_s(i_d^2+i_q^2)+\omega_r(L_q-L_d)i_qi_d-\omega_r\lambda_mi_q}
\end{array}\\
 \;\;\dot{e}_\txt{int} = -v_\txt{dc} + v_\dc^*
\end{array}
\end{align}
For simplicity, \eqref{eq:overalldc} is written as the  nonlinear system:
\begin{align}
\dot{x} &= f_c\pa{x,\;u,\;d}
\end{align}
where $x = \pa{v_\dc,\;e_\txt{int}}^T$, $u=\pa{i_d,\;i_q}^T$, $d=i_L$.  

The extended nonlinear system \eqref{eq:overalldc} is then discretized at a certain time step $T_s$:
\begin{align}
	x_{k+1} = f_d(x_k,\;u_k,\;d_k)
\end{align}
using Forward Euler (FE). The NMPC can now be formally stated:
\begin{align}\label{eq:NMPCprob}
\begin{array}{l}
	\displaystyle \min_{x_k,\;u_k}\q \sum_{k=0}^{N-1} (x_k-x_\txt{ref})^TQ(x_k-x_\txt{ref}) + u_k^T R u_k+ \\
	\q\q\q (x_N-x_\txt{ref})Q(x_N-x_\txt{ref})  \\
	\txt{s.t.} \\
	\left\{\begin{array}{l}
		\q x_{k+1} = f_d(x_k,\;u_k,\;d_k)  \\
		\q ||u_k||_2^2\leq I_\txt{peak}^2 \q\txt{for }k=0,\;...,\;N-1\\ 
		\q \pa{\omega_rL_qu_{2,k}}^2+\pa{\omega_rL_du_{1,k}+\omega_r\lambda_m}^2\leq  \pa{\dfrac{x_{1,k}}{2}}^2  
	\end{array}\right.\\
	\q V_\txt{dc-min}\leq x_{1,k}\leq V_\txt{dc-max} \q \txt{for } k= 1,\;...,\;N
	\end{array}
\end{align}
where $N$ is the prediction horizon, $x_\txt{ref} = \pa{v_\dc^*,\;0}^T$, and $Q,\;R\succ 0$. The main advantage of using the proposed NMPC is that under certain conditions, the optimal solution to \eqref{eq:NMPCprob} satisfies \eqref{eq:optstatic} during steady state, as shown in the following proposition.
\begin{proposition}\label{eq:thm1}
	Assume $Q=\txt{Diag}\pa{\gamma,\;\beta}$ and $R = \gamma I$, where $\epsilon,\;\gamma,\;\beta$ are positive constants. Let $P_e \teq -\pa{v_{dc}^2/R+v_{dc}i_{L}}$, $r_s = 0$, and perfect tracking is achieved, i.e. $x_{1,k} = v_\dc^*$ for $k$ greater than a certain $M$.
	
	During steady state ($x_{k+1} = x_k$), assume the optimal solution to \eqref{eq:NMPCprob} is as follows: $X^* = x_k^*\otimes\mbf{1}_N^T$ and $U^*\teq u_k^*\otimes\mbf{1}_N^T$, where $N$ is the horizon. Then $u_k^*$ is also a solution to \eqref{eq:optstatic}. 
\end{proposition}
\begin{proof}
	During steady state, the optimal solution of the MPC problem satisfies:
	\begin{align}
		x_k^*= x_k^*+T_sf_c(x_k^*, \;u_k^*,\;d_k)\;\;\txt{(using FE)}
	\end{align}
	Using \eqref{eq:overalldc}, the previous equation simplifies to:
	\begin{align}\label{eq:ssproof}
		-\pa{\frac{v_{dc}}{R}+i_{L}}  = \frac{3}{2}\frac{1}{v_{dc}}\pa{\omega_r(L_d-L_q)i_qi_d+\omega_r\lambda_mi_q}
	\end{align}
	Multiplying both sides of \eqref{eq:ssproof} by $v_{dc}$ we obtain:
	\begin{align}
		P_e= -\pa{\frac{v_{dc}^2}{R}+v_{dc}i_{L}}  = \frac{3}{2}\omega_r\pa{\lambda_mi_q+(L_d-L_q)i_qi_d}
	\end{align}
	Therefore,  the same equality constraint of \eqref{eq:optstatic} is obtained by the previous equation. Lastly, since $R = \gamma I$ implies that $u_k^TRu_k = \gamma||u_k||_2^2$, during steady state the cost function (besides $\gamma$) and constraints of \eqref{eq:NMPCprob} are equivalent to \eqref{eq:optstatic}. Therefore, the solution $u_k^*$ for \eqref{eq:NMPCprob} during steady state is also a solution to \eqref{eq:optstatic}.
\end{proof}

\begin{table}[!b]\footnotesize
	\renewcommand{\arraystretch}{1.5}\vspace{10pt}
	\caption{PMSM parameters based on the BMW i3 motor/generator \cite{Dajaku, ozpineci2016oak}.}
	\label{tab:machpars}
	\centering
	\begin{tabular}{l c l c l c }
		\hline \hline
		$L_d$ 			   		& 0.090 mH 		& $L_q$ 		& 0.255 mH 	&$ \lambda_m$ 	   		& 0.0385 Vs 	\\ \hline
		$r_s$ 		& 5.3 m$\Omega$  &	$n_\txt{max}$ 	   		& 11400 rpm 	& Poles 		& 12 			\\ \hline
		$T_\txt{max}$ 	   		& 250 Nm 		& $P_\txt{max}$ & 125 kW 	& $I_\txt{phase-peak}$ 	& 400 A				\\ 
		\hline \hline
	\end{tabular}
\end{table}

\begin{figure}[!b]
	\begin{center}
		\includegraphics[width=0.45\textwidth]{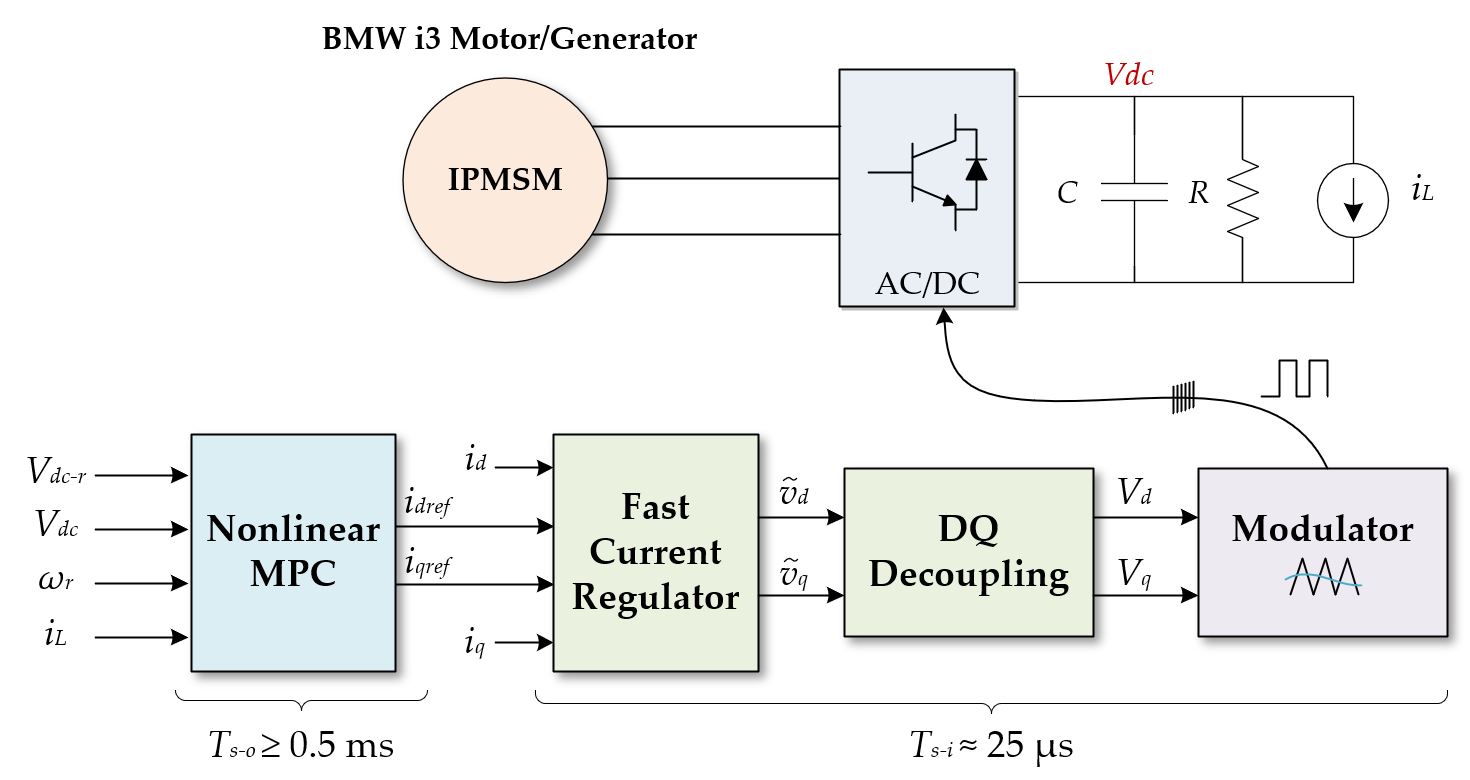}\vspace{0pt}
		\caption{Proposed control strategy for a IPMSM based generator system in dc microgrids. } 
		\label{fig:specificctrl} 
	\end{center}
\end{figure}

The proposed control strategy not only  dynamically regulates $v_\txt{dc}$ to the reference voltage, but also optimizes the steady state based on \eqref{eq:optstatic}.  During high speed operation, it may not always be possible to be in the interior of $\mc{E}$ and flux weakening is implicitly achieved by ensuring the current and voltage limits in $\mc{E}$ are satisfied.  

\renewcommand{\meas}{.45}
\begin{figure}[!t]
	\centering
	\subfloat[$DQ$ currents (top) and dc bus voltage (bottom). The dashed lines represent the references.]{\label{fig:Dq_case1}\includegraphics[width=\meas\textwidth]{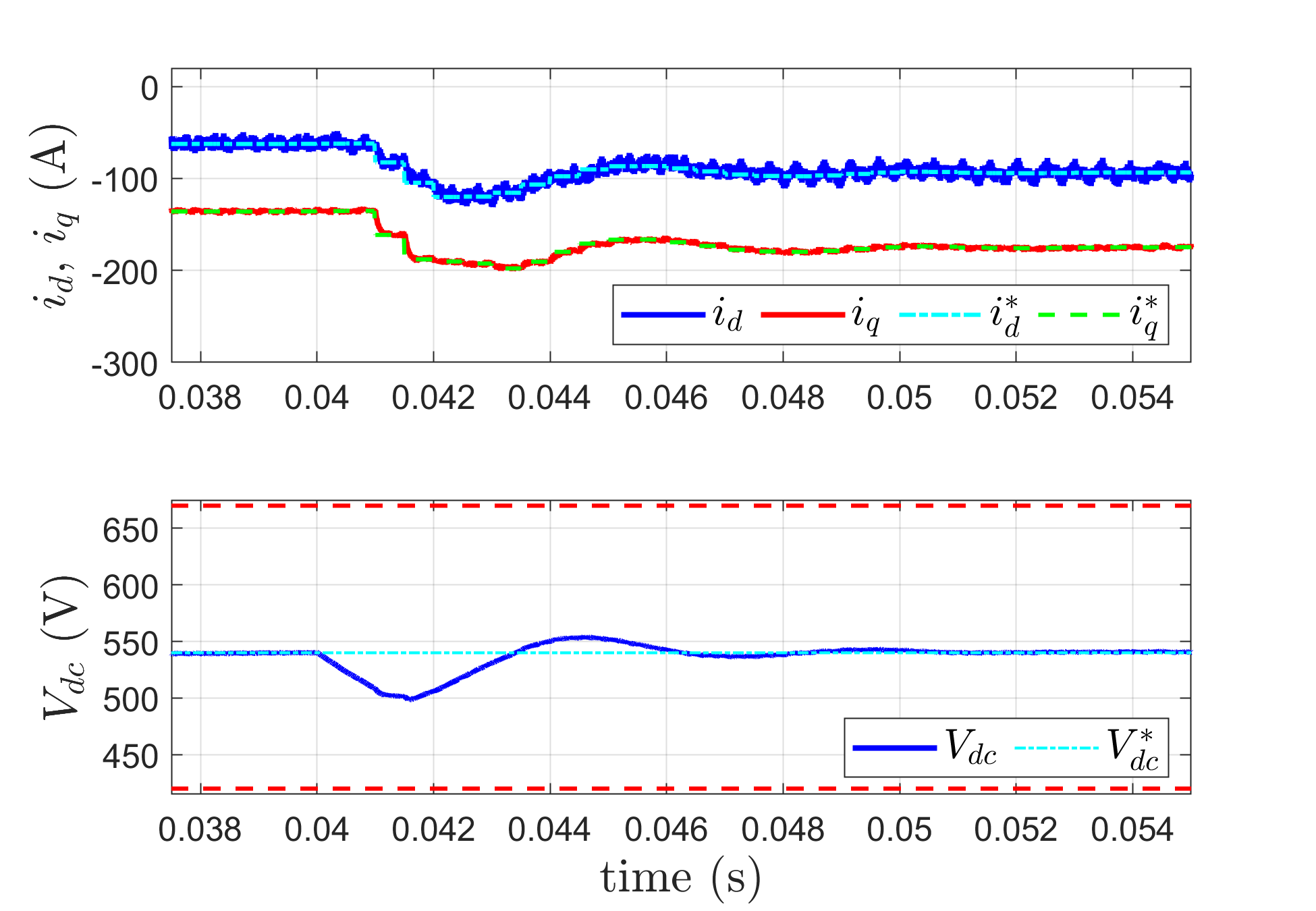}} \\
	\subfloat[Three phase currents (top) and modulation indices in $abc$ form (bottom).]{\label{fig:Iabc_case1}\includegraphics[width=\meas\textwidth]{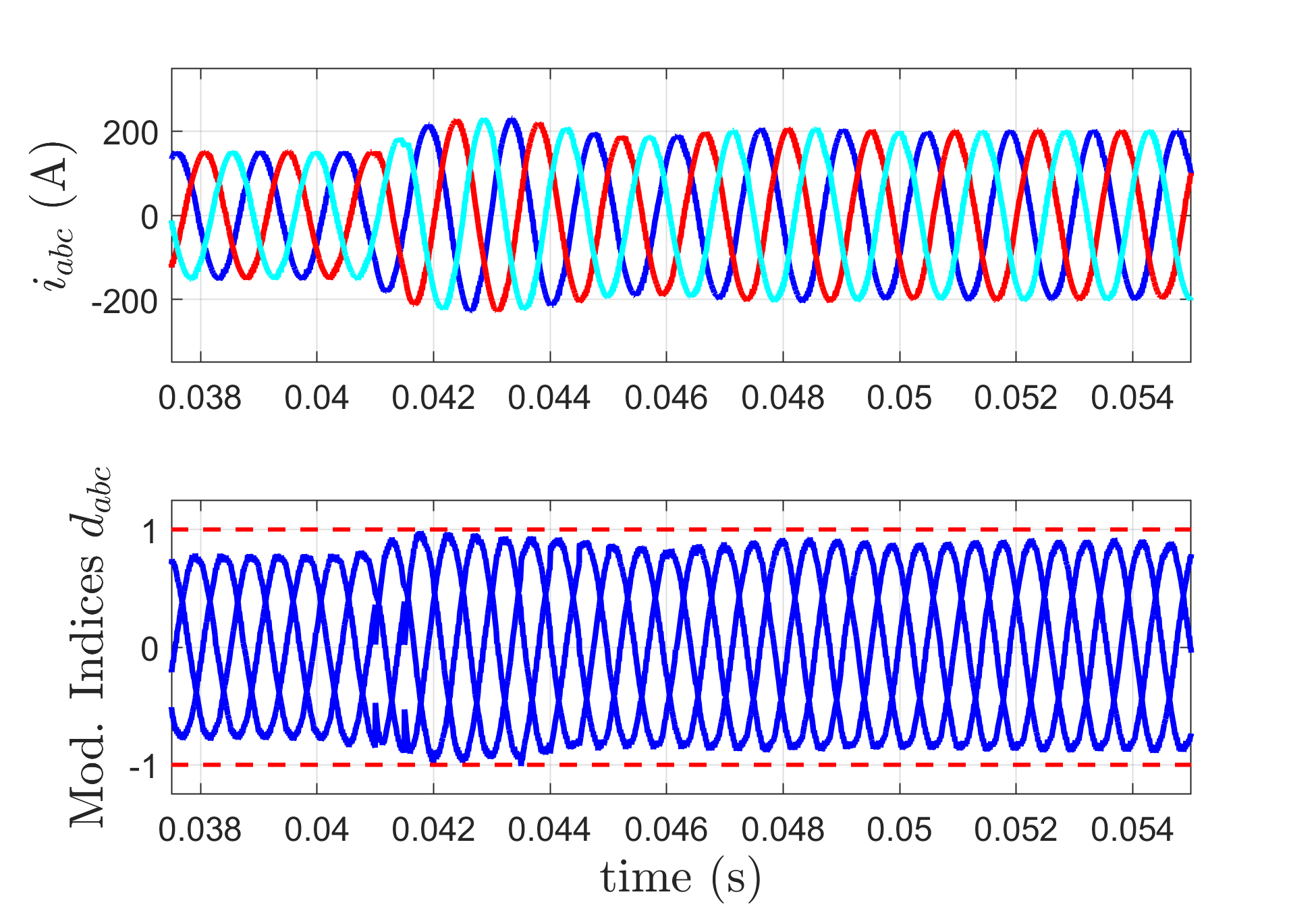}}
	\caption{Simulation results for case 1. A load change from 43.5 kW to 62.25 kW occurs at $t=0.04$ s.}\label{fig:case1res}
\end{figure}

\begin{table}[!t]\footnotesize
	\renewcommand{\arraystretch}{1.5}\vspace{10pt}
	\caption{Control parameters for active rectification of an IPMSM based generator.}
	\label{tab:ctrlpars}
	\centering
	\begin{tabular}{l c l c}
		\hline \hline
		$F_{\sw}$ (switching freq.)	& 40 kHz 			& $T_{s-o}$ 		& 0.5 ms 		\\ \hline
		$T_{s-i}$   	   				& $25\;\mu\txt{s}$ 	& $V_\txt{dc-min}$ 	& 420 V			\\ \hline
		$V_\txt{dc-max}$ 	   			& 670 V 			& N     			& 10 			\\ \hline
		$Q$ 			   				& $\txt{Diag}\pa{0.1,\;9000}$     	& $R$ 	& $0.1I_{2\times 2}$	\\
		\hline \hline
	\end{tabular}
\end{table}

\section{Case Study and Simulation Results}
We consider the parameters for the IPMSM shown in Tab. \ref{tab:machpars}.  These machine parameters are based on the BMW i3 motor/generator \cite{Dajaku, ozpineci2016oak}.  The voltage reference is set to 540 V with a maximum load of 125 kW. The NMPC discretization rate is at least $T_{s-o} \geq 0.5$ ms while the inner loop current regulator sampling time is $T_{s-i} = 25\;\mu\txt{s}$ (corresponding to a $F_\sw=40$ kHz switching frequency). The overall control strategy is shown in Fig. \ref{fig:specificctrl}. As can be seen in this figure, the NMPC is the outer  control associated with the slow subsystem ($v_{\dc}$), with inputs as the reference $dq$ currents to be used in the fast current regulator.  The $dq$ decoupling block is based on equations \eqref{eq:ddecoup} and \eqref{eq:qdecoup}. Finally, the modulator uses \eqref{eq:modinx} to compute the modulation indices for sine PWM.  The control parameters are summarized in Tab. \ref{tab:ctrlpars}.

\begin{figure}[!t]
	\centering
	\subfloat[$DQ$ currents (top) and dc bus voltage (bottom). The dashed lines represent the references.]{\label{fig:Dq_case2}\includegraphics[width=\meas\textwidth]{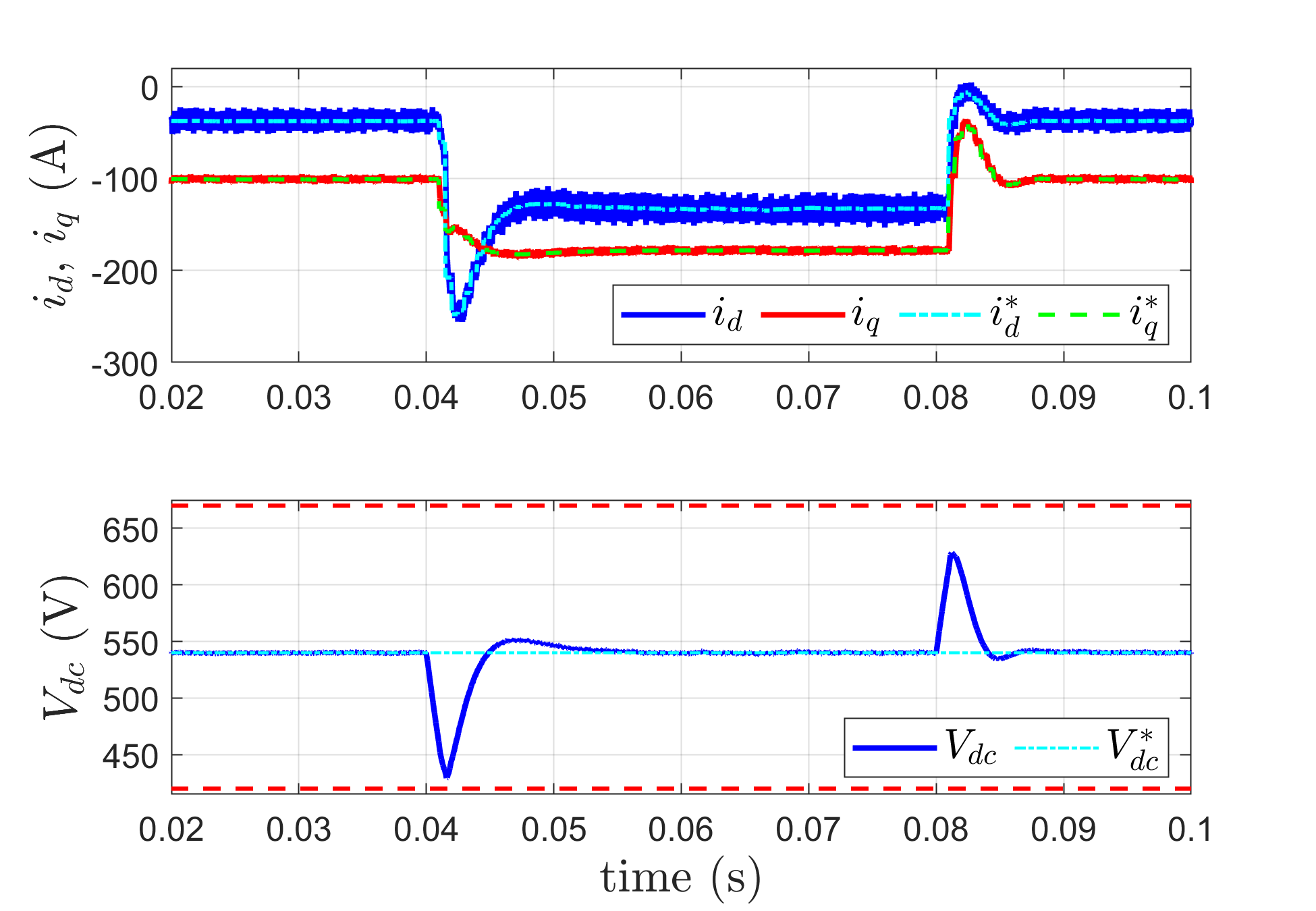}}\\
	\subfloat[Three phase currents (top) and modulation indices in $abc$ form (bottom).]{\label{fig:Iabc_case2}\includegraphics[width=\meas\textwidth]{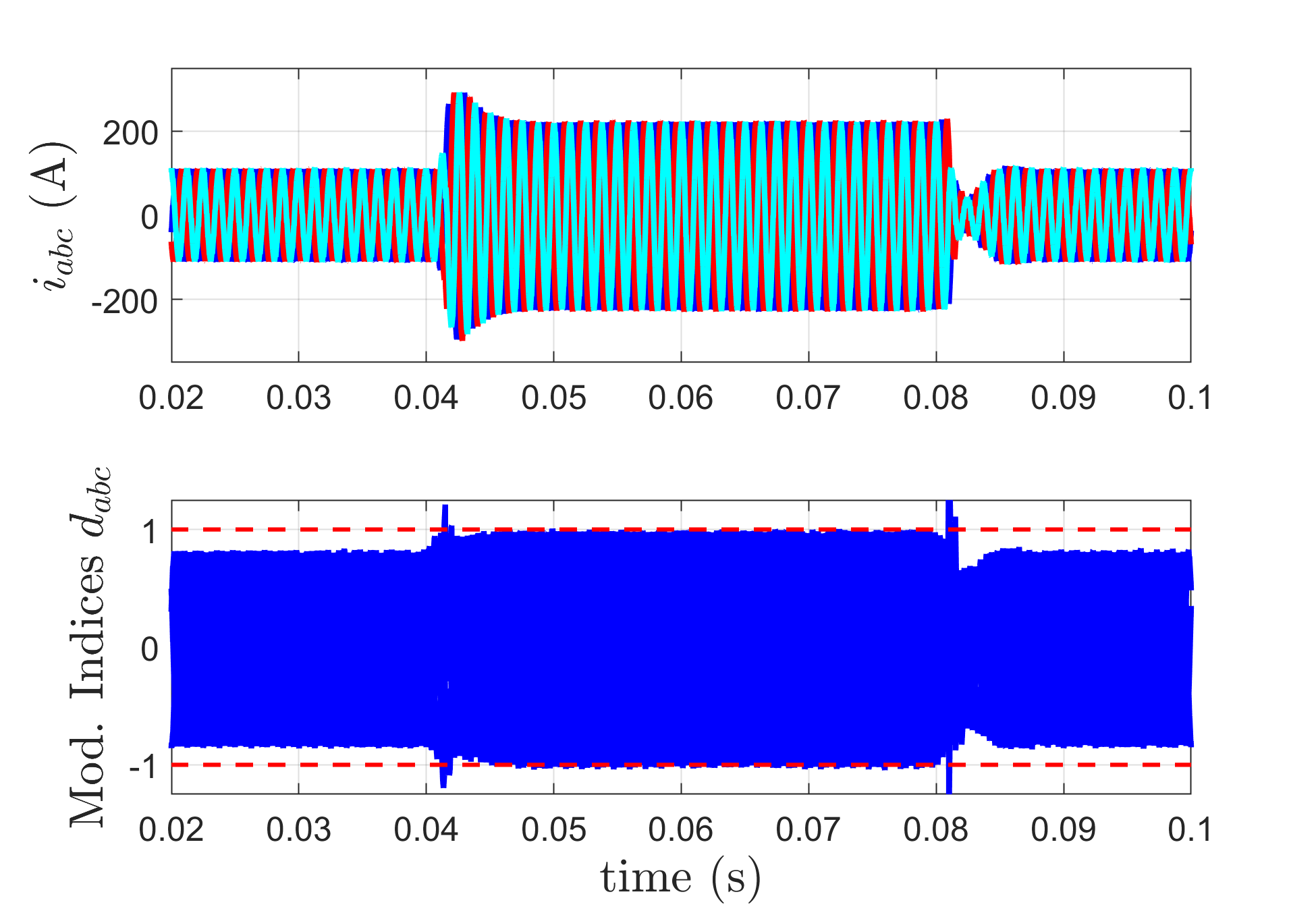}}
	\caption{Simulation results for case 2. A pulsed load occurs at $t=0.04$ s (on) and $t=0.08$ s (off). The load changes from $34$ kW to $81$ kW.}\label{fig:case2res}
\end{figure}

\subsection{Case 1}
We first consider the parameters in Tab. \ref{tab:ctrlpars} with a dc load change from 43.5 kW to 62.25 kW at $t = 0.04$ s.  The mechanical speed of the machine is $n=7000$ rpm. The optimal currents at 43.5 kW  can be solved using \eqref{eq:optstatic} as $i_\txt{d-opt} = -62$ A and $i_\txt{q-opt}= -135.3$, while at 62.25 kW are  $i_\txt{d-opt} = -93.5$ A and $i_\txt{q-opt}= -174.9$ A. The steady state values for the $dq$ currents are optimal for both of these power levels, as can be seen from Fig. \ref{fig:Dq_case1}. In addition, the reference currents are tracked accurately and much faster than the NMPC sampling time. The dc bus voltage is regulated  within 10 ms and is maintained within the bounds (dashed red). 

Fig \ref{fig:Iabc_case1} shows the phase currents and the modulation indices (abc).  Both of these can be obtained using the inverse Park transformation:
\begin{align}
	I_\txt{abc} &= K^{-1}\pa{i_d,\;i_q,\;0}^T \;\txt{and}\;
	d_\txt{abc} = K^{-1}\pa{d_d,\;d_q,\;0}^T
\end{align}
where $K$ is defined in the appendix. As mentioned previously, $d_d,\;d_q\in\br{-1,1}$ for sine PWM, which is implicitly enforced through \eqref{eq:voltineq}.

\subsection{Case 2}
Next, we consider a pulsed load change from 34 kW to 81 kW at $t = 0.04$ s and $t = 0.08$ s (on/off respectively). The mechanical speed in this case is $n=8000$ rpm. The control parameters are the same as the previous case. Fig. \ref{fig:Dq_case2} shows the $dq$ currents and the dc bus voltage. It can be seen that the voltage is kept within its limits and converges to the reference of $540$ V. Fig. \ref{fig:Iabc_case2} shows the phase currents and the modulation signals. It can be seen that when the load is set to $81$ kW, the modulation indices reach their limit of $\pm 1$. This implies that inequality \eqref{eq:voltineq} is binding at this load power.  In this mode of operation, more $i_d$ current is added to reduce the effect of the permanent magnet flux linkage and its induced back emf. 

\section{Conclusion and Future Work}
A controller design is presented for PMSG in dc microgrids. The proposed method is analyzed using similar assumptions of singular perturbation theory.  The inner loop controller for the ac currents is developed using output regulation while the outer loop control for the dc bus voltage tracking is based on NMPC.  It is shown that the NMPC is able to track the dc bus voltage accurately and minimize the peak ac currents, increasing efficiency.  Simulation results are presented using parameters for the BMW i3 IPMSM. Future work includes full hardware testing of the proposed controller and stability analysis of the proposed techniques. 

\section{Acknowledgement}
This research was supported by the AFRL Summer Faculty Fellowship Program (SFFP). Distribution A: approved for public release, distribution unlimited. Case Number: 88ABW-2020-2970. 

\appendix
The abc to dq transformation used in the derivation of \eqref{eq:dcside} and \eqref{eq:acside} is the following:
\begin{align}
	K = \frac{2}{3}\pmt{\cos(\theta_r) \s \cos(\theta_r-2\pi/3) \s \cos(\theta_r+2\pi/3) \\
										-\sin(\theta_r) \s -\sin(\theta_r-2\pi/3) \s -\sin(\theta_r+2\pi/3) \\
									    \frac{1}{2} \s \frac{1}{2} \s \frac{1}{2}} 
\end{align}

\bibliographystyle{abbrv}
\bibliography{IEEEabrv,bibfile}

\begin{thebibliography}{10}

\bibitem{704fstd}
{\em {MIL-STD-704F}, Aircraft Electric Power Characteristics}.
\newblock Military Standard.

\bibitem{Bozhko2017}
S.~{Bozhko}, M.~{Rashed}, C.~I. {Hill}, S.~S. {Yeoh}, and T.~{Yang}.
\newblock Flux-weakening control of electric starter–generator based on
  permanent-magnet machine.
\newblock {\em IEEE Transactions on Transportation Electrification},
  3(4):864--877, 2017.

\bibitem{Chau}
K.~T. {Chau}, C.~C. {Chan}, and C.~{Liu}.
\newblock Overview of permanent-magnet brushless drives for electric and hybrid
  electric vehicles.
\newblock {\em IEEE Transactions on Industrial Electronics}, 55(6):2246--2257,
  2008.

\bibitem{Clements2009}
N.~{Clements}, G.~{Venkataramanan}, and T.~M. {Jahns}.
\newblock Design considerations for a stator side voltage regulated permanent
  magnet ac generator.
\newblock In {\em 2009 IEEE Energy Conversion Congress and Exposition}, pages
  2763--2770, 2009.

\bibitem{Dajaku}
G.~{Dajaku}, H.~{Zhou}, X.~{Dajaku}, and D.~{Gerling}.
\newblock Novel rotor design with reduced rare-earth material for pm machines.
\newblock In {\em 2019 IEEE International Electric Machines Drives Conference
  (IEMDC)}, pages 1--7, 2019.

\bibitem{Dehghani}
H.~{Dehghani Tafti}, A.~I. {Maswood}, Z.~{Lim}, G.~H.~P. {Ooi}, and P.~H.
  {Raj}.
\newblock Proportional-resonant controlled npc converter for
  more-electric-aircraft starter-generator.
\newblock In {\em 2015 IEEE 11th International Conference on Power Electronics
  and Drive Systems}, pages 41--46, 2015.

\bibitem{Fan2018}
L.~{Fan}, T.~{Yang}, M.~{Rashed}, and S.~{Bozhko}.
\newblock Sensorless control of dual-three phase pmsm based aircraft electric
  starter/generator system using model reference adaptive system method.
\newblock In {\em CSAA/IET International Conference on Aircraft Utility Systems
  (AUS 2018)}, pages 787--794, 2018.

\bibitem{Francis}
B.~A. Francis.
\newblock The linear multivariable regulator problem.
\newblock {\em SIAM Journal on Control and Optimization}, 15(3):486--505, 1977.

\bibitem{Gao2016}
F.~{Gao} and S.~{Bozhko}.
\newblock Modeling and impedance analysis of a single dc bus-based
  multiple-source multiple-load electrical power system.
\newblock {\em IEEE Transactions on Transportation Electrification},
  2(3):335--346, 2016.

\bibitem{Gao1}
F.~{Gao}, X.~{Zheng}, S.~{Bozhko}, C.~I. {Hill}, and G.~{Asher}.
\newblock Modal analysis of a pmsg-based dc electrical power system in the more
  electric aircraft using eigenvalues sensitivity.
\newblock {\em IEEE Transactions on Transportation Electrification},
  1(1):65--76, 2015.

\bibitem{Giangrande}
P.~{Giangrande}, V.~{Madonna}, G.~{Sala}, A.~{Kladas}, C.~{Gerada}, and
  M.~{Galea}.
\newblock Design and testing of pmsm for aerospace ema applications.
\newblock In {\em IECON 2018 - 44th Annual Conference of the IEEE Industrial
  Electronics Society}, pages 2038--2043, Oct 2018.

\bibitem{HerreraNCS}
L.~{Herrera}, E.~{Inoa}, F.~{Guo}, J.~{Wang}, and H.~{Tang}.
\newblock Small-signal modeling and networked control of a phev charging
  facility.
\newblock {\em IEEE Transactions on Industry Applications}, 50(2):1121--1130,
  2014.

\bibitem{HerreraTSG}
L.~Herrera, W.~Zhang, and J.~Wang.
\newblock Stability analysis and controller design of dc microgrids with
  constant power loads.
\newblock {\em IEEE Transactions on Smart Grid}, 8(2):881--888, March 2017.

\bibitem{kimball2008singular}
J.~W. Kimball and P.~T. Krein.
\newblock Singular perturbation theory for dc--dc converters and application to
  pfc converters.
\newblock {\em IEEE Transactions on Power Electronics}, 23(6):2970--2981, 2008.

\bibitem{kokotovic1999singular}
P.~Kokotovi{\'c}, H.~K. Khalil, and J.~O'reilly.
\newblock {\em Singular perturbation methods in control: analysis and design}.
\newblock SIAM, 1999.

\bibitem{Miao}
D.~{Miao}, Y.~{Mollet}, J.~{Gyselinck}, and J.~{Shen}.
\newblock Dc voltage control of a wide-speed-range permanent-magnet synchronous
  generator system for more electric aircraft applications.
\newblock In {\em 2016 IEEE Vehicle Power and Propulsion Conference (VPPC)},
  pages 1--6, 2016.

\bibitem{nam2018ac}
K.~H. Nam.
\newblock {\em AC motor control and electrical vehicle applications}.
\newblock CRC press, 2018.

\bibitem{ozpineci2016oak}
B.~Ozpineci.
\newblock Oak ridge national laboratory annual progress report for the electric
  drive technologies program.
\newblock Technical report, Oak Ridge National Lab.(ORNL), Oak Ridge, TN
  (United States)., 2016.

\bibitem{pogaku2007modeling}
N.~Pogaku, M.~Prodanovic, and T.~C. Green.
\newblock Modeling, analysis and testing of autonomous operation of an
  inverter-based microgrid.
\newblock {\em IEEE Transactions on power electronics}, 22(2):613--625, 2007.

\bibitem{tripathi2016optimum}
S.~M. Tripathi, A.~N. Tiwari, and D.~Singh.
\newblock Optimum design of proportional-integral controllers in
  grid-integrated pmsg-based wind energy conversion system.
\newblock {\em International Transactions on Electrical Energy Systems},
  26(5):1006--1031, 2016.

\bibitem{umbria2014three}
F.~Umbr{\'\i}a, J.~Aracil, F.~Gordillo, F.~Salas, and J.~A. S{\'a}nchez.
\newblock Three-time-scale singular perturbation stability analysis of
  three-phase power converters.
\newblock {\em Asian Journal of Control}, 16(5):1361--1372, 2014.

\bibitem{vasquez2012modeling}
J.~C. Vasquez, J.~M. Guerrero, M.~Savaghebi, J.~Eloy-Garcia, and R.~Teodorescu.
\newblock Modeling, analysis, and design of stationary-reference-frame
  droop-controlled parallel three-phase voltage source inverters.
\newblock {\em IEEE Transactions on Industrial Electronics}, 60(4):1271--1280,
  2012.

\bibitem{wang2014modeling}
X.~Wang, F.~Blaabjerg, and W.~Wu.
\newblock Modeling and analysis of harmonic stability in an ac
  power-electronics-based power system.
\newblock {\em IEEE Transactions on Power Electronics}, 29(12):6421--6432,
  2014.

\end{thebibliography}
\par
\leavevmode

\end{document}